\documentclass[11pt]{article}
\usepackage[utf8]{inputenc}

\usepackage{amsmath,amssymb,amsfonts,amsthm}
\usepackage{libertine}       % Sets the main text font to Libertine
\usepackage[libertine]{newtxmath} % Sets the math font to match Liberti

\usepackage[linesnumbered,ruled,vlined]{algorithm2e}
\usepackage{algorithmic}
\usepackage[margin=1.5in]{geometry}
\usepackage{tikz}
\usepackage{float}
\usepackage{algorithmic}

\allowdisplaybreaks

\usepackage{url}
\usepackage{hyperref}
\usepackage[hyphenbreaks]{breakurl}

\usepackage{multicol,
			graphicx,
			amsmath,
                amsthm,
			hyperref,
			changepage,
                blindtext,
                titlesec,
                amssymb,
                enumerate,
                mathtools,
			tikz}
\newtheorem{lemma}{Lemma}

\newtheorem{theorem}{Theorem}

\newtheorem{fact}{Fact}

\newcommand{\D}{\displaystyle}

\newcommand{\RR}{\mathbb{R}}

\newcommand{\Opt}{\textup{OPT}}

\newcommand{\rounddown}[1]{\lfloor #1 \rfloor}
\newcommand{\roundup}[1]{\lceil #1 \rceil}
\newcommand{\tO}[1]{\tilde{O} \left(  #1 \right)}

\renewcommand{\exp}[1]{\text{exp}\left( #1 \right)}

\DeclareMathOperator{\poly}{poly}

\title{Massively Parallel Maximum Coverage Revisited \footnote{The conference version of this manuscript is to appear in the 50th International Conference on Current Trends in Theory and Practice of Computer Science.}}
\author{Thai Bui\thanks{tbui8182@sdsu.edu, San Diego State University. Supported by NSF Grant No. 2342527.} ,  Hoa T. Vu\thanks{hvu2@sdsu.edu, San Diego State University. Supported by NSF Grant No. 2342527.}}
\date{}

\begin{document}

\maketitle

\begin{abstract}
    We study the maximum set coverage problem in the massively parallel model. In this setting, $m$ sets that are subsets of a universe of $n$ elements are distributed among $m$ machines. In each round, these machines can communicate with each other, subject to the memory constraint that no machine may use more than $\tO{n}$ memory. The objective is to find the $k$ sets whose coverage is maximized. We consider the regime where $k = \Omega(m)$  (i.e., $k = m/100$), $m = O(n)$, and each machine has $\tO{n}$ memory \footnote{The input size is $O(mn)$ and each machine has the memory enough to store a constant number of sets.}.

    Maximum coverage is a special case of the submodular maximization problem subject to a cardinality constraint. This problem can be approximated to within a $1-1/e$ factor using the greedy algorithm, but this approach is not directly applicable to parallel and distributed models. When $k = \Omega(m)$, to obtain a $1-1/e-\epsilon$ approximation, previous work either requires $\tO{mn}$ memory per machine which is not interesting compared to the trivial algorithm that sends the entire input to a single machine, or requires $2^{O(1/\epsilon)} n$ memory per machine which is prohibitively expensive even for a moderately small value $\epsilon$.

    Our result is a randomized $(1-1/e-\epsilon)$-approximation algorithm that uses 
    \[O(1/\epsilon^3 \cdot \log m \cdot (\log (1/\epsilon) + \log m))\] rounds. Our algorithm involves solving a slightly transformed linear program of the maximum coverage problem using the multiplicative weights update method, classic techniques in parallel computing such as parallel prefix, and various combinatorial arguments. 

\end{abstract}

\section{Introduction}

Maximum coverage is a classic NP-Hard problem.  In this problem, we have $m$ sets $S_1,S_2,\ldots,S_m$ that are subsets of a universe of $n$ elements $[n] = \{1,2,\ldots,n\}$. The goal is to find $k$ sets that cover the maximum number of elements. In the offline model, the greedy algorithm achieves a $1-1/e$ approximation and assuming $\textup{P} \neq \textup{NP}$, this approximation is the best possible in polynomial time \cite{Feige98}.

However, the greedy algorithm for maximum coverage and the related set cover problem  is not friendly to  streaming, distributed, and massively parallel computing. A large body of work has been devoted to designing algorithms for these problems in these big data computation models. An incomplete list of work includes \cite{MV19,IndykV19,MTV21,Khanna0A23,AssadiKL21,IndykMRUVY17,AssadiK18,ChakrabartiMW24,Har-PeledIMV16,CervenjakGUW24,WarnekeCW23,SahaG09,Assadi17,JaudWC23}. 

Some example applications of maximum coverage includes facility and sensor placement \cite{KrauseG07}, circuit layout and job scheduling \cite{hochbaum1998analysis}, information retrieval \cite{Anagnostopoulos15}, market design \cite{KempeKT15}, data summarization \cite{SahaG09}, and social network analysis \cite{JaudWC23}.

\paragraph{The MPC model.} We consider the massively parallel computation model (MPC) in which $m$ sets $S_1,S_2,\ldots,S_m \subseteq [n]$ are distributed among $m$ machines. Each machine has memory $\tilde{O}(n)$ and holds a set. In each round, each machine can communicate with others with the constraint that no machine receives a total message of size more than $\tilde{O}(n)$. Similar to previous work in the literature, we assume that $m \leq n$. 

The MPC model, introduced by Karloff, Suri, and Vassilvitskii \cite{KarloffSV10} is an abstraction of various modern computing paradigms such as MapReduce, Hadoop, and Spark. 

\paragraph{Previous work.} This problem is a special case of submodular maximization subject to a cardinality constraint. The results of Liu and Vondrak \cite{LiuV19}, Barbosa et al. \cite{BarbosaENW16}, Kumar et al. \cite{KumarMVV15} typically require that each machine has enough memory to store $O(\sqrt{km})$  items which are sets in our case (and storing a set requires $\tO{n}$ memory) with $\sqrt{m/k}$ machines. When $k = \Omega(m)$ (e.g., $k = m/100$), this means that a single machine may need $\tO {m n}$ memory. This is not better than the trivial algorithm that sends the entire input to a single machine and solves the problem in 1 round.

Assadi and Khanna gave a randomized $1-1/e-\epsilon$ approximation algorithm in which each machine has  $\tO{m^{\delta/\epsilon} n}$ memory and the number of machines is $m^{1-\delta/\epsilon}$ for any $\epsilon, \delta \in (0,1)$  (see Corollary 10 in the full paper of \cite{AssadiK18}). Setting $\delta = \Theta(1/\log m)$ gives us a $1-1/e-\epsilon$ approximation in $O(1/\epsilon \cdot \log m)$ rounds with $O(m)$ machines each of which uses $\tO{2^{1/\epsilon} n}$ memory. While Assadi and Khanna's result is nontrivial in this regime, the dependence on $\epsilon$ is exponential and if $n$ is large, then even a moderately small value of $\epsilon=0.01$ can lead to a prohibitively large memory requirement $\approx 2^{100}n$. Their work however can handle the case where $k = o(m)$.

\paragraph{Our result.} We present a relatively simple randomized algorithm that achieves a $1-1/e-\epsilon$ approximation in  $O(1/\epsilon^3 \cdot \log m \cdot (\log (1/\epsilon) + \log m))$ rounds with  $\tO{n}$ memory per machine assuming $k = \Omega(m)$. Our space requirement does not depend on $\epsilon$ compared to the exponential dependence in Assadi and Khanna's result. 

We note that assuming $k = \Omega(m)$ does not make the problem any easier since there are still exponentially many solutions to consider. In practice, one can think of many applications where one can utilize a constant fraction of the available sets (e.g., $10\%$ or  $20\%$). We state our main result as a theorem below.

\begin{theorem}\label{thm:main}
    Assume $k = \Omega(m)$ and there are $m$ machines each of which has $\tO{n}$ memory. There exists an algorithm that with high probability finds $k$ sets that cover at least $(1-1/e-\epsilon)\Opt$ elements in $O(1/\epsilon^3 \cdot \log m \cdot (\log (1/\epsilon) + \log m))$ rounds.
\end{theorem}

If the maximum frequency $f$ (the maximum number of sets that any element belongs to) is bounded, we can drop the assumption that $k= \Omega(m)$, and parameterize the number of rounds based on $f$. In particular, we can obtain a $1-1/e-\epsilon$ approximation in $O(f^3/\epsilon^6 \cdot \log^2 (\frac{kf}{\epsilon}))$ rounds.

\paragraph{Remark.} We could easily modify our algorithm so that each machine uses $\tO{1/\epsilon \cdot n}$ memory and the number of rounds is $O(1/\epsilon^2 \cdot   \log m \cdot (\log (1/\epsilon) + \log m))$.  At least one  $\log m$ factor is necessary based on the lower bound given by Corollary 9 of \cite{AssadiK18}. 

Randomization appears in two parts of our algorithms: the rounding step and the subsampling step to reduce the number of rounds from $\log m \cdot \log n$ to $\log m \cdot (\log(1/\epsilon) + \log m)$. If we only need to compute an approximation to the optimal coverage value such that the output is in the interval $[(1-\epsilon)\Opt,\Opt/(1-1/e-\epsilon)]$, then we have a deterministic algorithm that runs in $O(1/\epsilon^3 \cdot \log n \cdot \log m)$ rounds. The algorithm by Assadi and Khanna \cite{AssadiK18} combines the sample-and-prune framework with threshold greedy.  This strategy requires sampling sets. It is unclear how to derandomize their algorithm even just to compute an approximation to the optimal coverage value.

\paragraph{Our techniques and paper organization.} In Section \ref{sec:main-alg}, we transform the standard linear program for the maximum coverage problem into an equivalent packing linear program that can be solved ``approximately'' by the multiplicative weights update method. At a high level, the multiplicative weights update method gives us a fractional solution that is a $1-1/e-O(\epsilon)$ bi-criteria approximation where $(1+O(\epsilon))k$ ``fractional'' sets cover $(1-1/e-O(\epsilon))\Opt$ ``fractional'' elements. We then show how to find $k$ sets covering $(1-1/e-O(\epsilon))\Opt$ elements from this fractional solution through a combinatorial argument and parallel prefix.

Section \ref{sec:mwu} outlines the details to solve the transformed linear program in the MPC model. While this part is an adaptation of the standard multiplicative weights, an implementation in the MPC model requires some additional details such as the number of bits to represent the weights.  All missing proofs can be found in the appendix.

\paragraph{Preliminaries.} In this work, we will always consider the case where each machine has $\tO{n}$ memory and $m \leq n$. Without loss of generality, we may assume the non-central machine $j$ stores the set $S_j$. For each element $i \in [n]$, we use $f_i$ to denote the number of sets that $i$ is in. This is also referred to as the frequency of $i$. Assume each machine has $\tO{n}$ space. The vector $\textbf{f}$ can be computed in $O(\log m)$ rounds and broadcasted to all machines. Each machine $j$ starts with the characteristic vector $v_j \in \{0,1\}^n$ of the set $S_j$ that it holds. The vector $\textbf{f}$ is just the sum of the characteristic vectors of the sets. We can aggregate the vectors $\{v_j \}$ in $O(\log m)$ rounds using the standard binary tree aggregation algorithm. 

Since in this work, the dependence on $1/\epsilon$ is polynomial, an $\alpha - O(\epsilon)$ approximation can easily be translated to an $\alpha  - \epsilon$ approximation by scaling $\epsilon$ by a constant factor. We can also assume that $1/\epsilon < n/10$; otherwise, we can simulate the greedy algorithm in $O(1/\epsilon)$ rounds. For the sake of exposition, we will not attempt to optimize the constants in our algorithm and analysis.  Finally, in this work we consider $1-1/\poly(m)$ as a ``high probability''.  We use $[E]$ to denote the indicator variable of the event $E$ that is 1 if $E$ happens and 0 otherwise.

 %%%%%%%%%%%%%%%%%%%%%%%%%%%%%%%%%%%%%%%%%%%%%%%%%%%%%%%%%%%%%%%%%%%%%%%%%%%%%%%%%%%%%%%%
%%%%%%%%%%%%%%%%%%%%%%%%%%%%%%%%%%%%%%%%%%%%%%%%%%%%%%%%%%%%%%%%%%%%%%%%%%%%%%%%%%%%%%%%
%%%%%%%%%%%%%%%%%%%%%%%%%%%%%%%%%%%%%%%%%%%%%%%%%%%%%%%%%%%%%%%%%%%%%%%%%%%%%%%%%%%%%%%%
%%%%%%%%%%%%%%%%%%%%%%%%%%%%%%%%%%%%%%%%%%%%%%%%%%%%%%%%%%%%%%%%%%%%%%%%%%%%%%%%%%%%%%%%

\section{Algorithm}

\subsection{The main algorithm} \label{sec:main-alg}

\paragraph{Linear programming (re)formulation.} We first recall the relaxed linear program (LP) for the maximum coverage problem $\Pi_0$:

    \begin{align*}
                \text{maximize} \quad               & \sum_{i \in [n]} x_i                                   \\
                \text{(s.t.)\quad}      &  x_i \le \sum_{S_j \ni i} y_j      &   & \forall i \in [n]\\
                                        & \sum_{j \in [m]} y_j = k        &   &\\
                                        & x_i, y_j \in [0,1]                &   & \forall i \in [n], j \in [m].  
    \end{align*}
    
We first reformulate this LP and then approximately solve the new LP using the multiplicative weights update method \cite{AroraHK12}. For each $j \in [m]$, let $z_j := 1 -  y_j$. We have the following fact.

    \begin{fact}\label{fact:equiv-lp}
        For each $i \in [n]$, $x_i \le \D\sum_{S_j \ni i} y_j \iff x_i + \sum_{S_j \ni i} z_j \le \sum_{S_j \ni i} (y_j+z_j) = f_i.$
    \end{fact}

Note that if $\textbf{y} \in [0,1]^m$ and $\sum_{j} y_j = k$, then $\textbf{z} \in [0,1]^m$ and $\sum_{j} z_j = m-k$. Thus, it is not hard to see that the original LP is equivalent to the following LP which we will refer to as $\Pi_1$.

    \begin{align*}
            \text{maximize} \quad               & \sum_{i \in [n]} x_i                                   \\
            \text{(s.t.)\quad}      & \frac{x_i}{f_i} + \frac{1}{f_i}\cdot\sum_{S_j \ni i} z_j \le 1      &   & \forall i \in [n]\\
                                    & \sum_{j \in [m]} z_j = m-k        &   &\\
                                    & x_i, z_j \in [0,1]                &   & \forall i \in [n], j \in [m].  
    \end{align*}

In this section, we will assume the existence an MPC algorithm that approximately solves the linear program $\Pi_1$ in $O (1/\epsilon^3 \cdot \log n \cdot \log m)$ rounds. The proof will be deferred to Section \ref{sec:mwu}.
    
\begin{theorem}\label{thm:solving-mkc-lp}
    There is an algorithm that  finds $\textbf{x} \in [0,1]^n, \textbf{z} \in [0,1]^m$ such that 
    \begin{enumerate}
        \item $\sum_{i \in [n]} x_i \geq (1-\epsilon) \Opt$,
        \item $\sum_{j \in [m]} z_j = m-k$, and
        \item 
        \[
            \frac{x_i}{f_i} + \frac{1}{f_i}\cdot\sum_{S_j \ni i} z_j \le 1 + \epsilon \ \ \, \forall i \in [n]
        \]
    \end{enumerate}
     
    in $O(\epsilon^{-3} \log n \cdot \log m)$ rounds.
\end{theorem}

Let $\textbf{x}$ and $\textbf{z}$ be the be the output given by Theorem \ref{thm:solving-mkc-lp}. Then, let $\textbf{x}^{\prime} = \textbf{x}/(1+\epsilon) $, $\textbf{z}^{\prime} = \textbf{z}/(1+\epsilon) $, and $\textbf{y}^{\prime}=\mathbf{1} -\textbf{z}^{\prime}$. We have
    \begin{align}
        \sum_{i=1}^n x^{\prime}_i & = \frac{1}{1+\epsilon} \sum_{i=1}^n x_i \geq \frac{1-\epsilon}{1+\epsilon} \Opt > (1-4\epsilon)\Opt, \\
        \sum_{j=1}^m y^{\prime}_i  & = \sum_{j=1}^m \left( 1- \frac{z_j}{1+\epsilon} \right)  = m - \frac{m-k}{1+\epsilon} \leq m - (1-2\epsilon) (m-k) < k + 2\epsilon m, \\
        x_i' + & \sum_{S_j \ni i} z_j' \leq f_i \iff  x_i'  \leq \sum_{S_j \ni i} y_j', \ \  \forall i \in [n]  \text{, by Fact \ref{fact:equiv-lp}}.
    \end{align}

Thus, by setting $\textbf{x} \leftarrow \textbf{x}/(1+\epsilon) $, and $\textbf{y} \leftarrow \textbf{y}^{\prime}$, we have an approximate solution  $\textbf{x}  \in [0,1]^n, \textbf{y} \in [0,1]^n$ to the LP $\Pi_0$ such that 
\begin{align*}
    \sum_{i=1}^n x_i & \geq  (1-4\epsilon)\Opt, && \sum_{j=1}^m y_j \leq k + 2\epsilon m, \text{ and } \\
    x_i  & \leq \sum_{S_j \ni i} y_j, \forall i \in [n].
\end{align*}

We can then apply the standard randomized rounding to find a sub-collection of at most $k + 2\epsilon m$ sets that covers at least $(1-4\epsilon)\Opt$ elements. For the sake of completeness, we will provide the rounding  algorithm in the MPC model in the following lemma. The proof can be found in Appendix \ref{sec:omitted-proofs}.

\begin{lemma} \label{lem:rounding}
    Suppose $\textbf{x} \in [0,1]^n$ and $\textbf{y} \in [0,1]^m$ satisfy:
    \begin{enumerate}
        \item $\sum_{i \in [n]} x_i \geq L$,
        \item $x_i \leq \sum_{S_j \ni i} y_j$ for all $i \in [n]$,
        \item $\sum_{j \in [m]} y_j = k$,
        \item $x_i, y_j \in [0,1]$ for all $i \in [n]$ and $j \in [m]$.
    \end{enumerate}
    Then there exists a rounding algorithm that finds a sub-collection of $k$ sets that in expectation cover at least $(1-1/e)L$ elements in $O(1)$ round. To obtain a high probability guarantee, the algorithm requires $O(1/\epsilon \cdot \log m)$ rounds to find  $k$ sets that cover least $(1-1/e-O(\epsilon))L$ elements.
\end{lemma}

Applying Lemma \ref{lem:rounding} to $\textbf{x}$ and $\textbf{y}$ with $k+2\epsilon m$ in place of $k$, we obtain a sub-collection of at most $k + 2\epsilon m$ sets that covers at least $(1-1/e-O(\epsilon))\Opt$ elements. Since we assume that $k = \Omega(m)$, that means we have found $k + O(\epsilon) k$  sets that cover at least $(1-1/e-O(\epsilon))\Opt$ elements. The next lemma shows that we can find $k$ sets among these that cover at least $(1-1/e-O(\epsilon))\Opt$ elements. The proof is a combination of a counting argument and the well-known parallel prefix algorithm  \cite{LadnerF80}.

    \begin{algorithm}
        \caption{Parallel prefix coverage} \label{alg:prefix-coverage}
        Compute $|S_1|, |S_2 \setminus S_1|, |S_3 \setminus (S_1 \cup S_2)|, \ldots, |S_k \setminus (S_1 \cup \ldots \cup S_{k-1})|$ in $O(\log k)$ rounds. \\
        \SetKwFunction{PrefixCoverage}{PrefixCoverage} % Define function name
        \SetKwProg{Fn}{Function}{:}{end} % Setup the function block structure

        \Fn{\PrefixCoverage{$S_1, S_2, \ldots, S_k$}}{
            \tcp{Compute $|S_1|, |S_2 \cup S_1|, |S_3 \cup S_2 \cup S_1|, \ldots, |S_k \cup S_{k-1} \cup \ldots \cup S_1|$}
            \If{$k = 1$} {
                \Return $|S_1|$.
            }\Else {
                \tcp{Assume $k$ is even.}
                In one round, machine $2j - 1$ sends $S_{2j-1}$ to machine $2j$, then machine $2j$ computes  $Q_{j} =  S_{2j-1} \cup S_{2j}$. \\
                Run \PrefixCoverage{$Q_1, Q_2, \ldots, Q_{k/2}$} on machines $2, 4, 6, \ldots, k$. \\
                Machine $j$ now has $S_1 \cup S_2 \cup S_3 \cup \ldots \cup S_{j}$ for $j = 2, 4, 6, \ldots, k$. \\
                In one round, machine $j = 1,3,5,\ldots, k - 1$ communicates with machine $j-1$ which has $S_1 \cup S_2 \cup \ldots S_{j-1}$ and computes $S_1 \cup S_2 \cup \ldots \cup S_{j}$. \\

                \tcp{If $k$ is odd, run the above algorithm on $S_1, S_2, \ldots, S_{k-1}$ and then compute $S_1 \cup S_2 \cup \ldots \cup S_k$ in one round.}

                Each machine $j$ communicates with machine $j-1$ to compute $|S_1 \cup S_2 \cup \ldots \cup S_j|-|S_1 \cup S_2 \cup \ldots \cup S_{j-1}|$ in one round. \\
                
            }
        }
        
    \end{algorithm}

We rely on the following result which is a simulation of the parallel prefix in our setting.

\begin{lemma}\label{lem:prefix-coverage}
    Suppose there are $k$ sets and machine $j$ holds the set $S_j$. Then Algorithm \ref{alg:prefix-coverage} computes $|S_1|, |S_2 \setminus S_1|, |S_3 \setminus (S_1 \cup S_2)|, |S_4 \setminus(S_1 \cup S_2 \cup S_3)|, \ldots, |S_k \cup S_{k-1} \cup \ldots \cup S_1|$ in $O(\log k)$ rounds.
\end{lemma}

\begin{proof}
    We first show how to compute $(S_1), (S_1 \cup S_2), (S_1 \cup S_2 \cup S_3), \ldots, (S_1 \cup S_2 \cup \ldots \cup S_k)$ in $O(\log k)$ rounds where machine $j$ holds $S_1 \cup S_2 \cup \ldots \cup S_j$ at the end. Once this is done, machine $j$ can send $S_1 \cup S_2 \cup \ldots \cup S_j$ to machine $j+1$ and machine $j+1$ can compute $|S_1 \cup S_2 \cup \ldots \cup S_{j+1}| - |S_1 \cup S_2 \cup \ldots \cup S_j| = |S_{j+1} \setminus (S_1 \cup S_2 \cup \ldots \cup S_j)|$.

    The algorithm operates recursively. In one round, machine $2j - 1$ sends $S_{2j-1}$ to machine $2j$, then machine $2j$ computes  $Q_{j} =  S_{2j-1} \cup S_{2j}$. Assuming $k$ is even, the algorithm recursively computes $(Q_1), (Q_1 \cup Q_2), (Q_1 \cup Q_2 \cup Q_3), \ldots, (Q_{1} \cup Q_{2} \cup \ldots \cup Q_{k})$ on machines $2, 4, \ldots, k$. After recursion, machines with even indices $2j$ has the set $S_1 \cup S_2 \cup \ldots \cup S_{2j}$. Then, in one round, machines with odd indices $2j+1$ communicate with machine $2j$ to learn about $S_1 \cup S_2 \cup \ldots \cup S_{2j+1}$. If $k$ is odd, we just do the same on $S_1, S_2, \ldots, S_{k-1}$ and then compute $S_1 \cup S_2 \cup \ldots \cup S_k$ in one round.
    
    There are $O(\log k)$ recursion levels and therefore, the total number of rounds is $O(\log k)$. 
\end{proof}

\begin{lemma}\label{lemma:fraction-of-a-collection}
    Let $\mathcal{S} = \{S_1,\dots,S_r\}$ be a collection of $r=(1+\gamma)k$ sets whose union contains $L$ elements where $\gamma \in [0,1)$, then there exist $k$ sets in $\mathcal{S}$ whose union contains at least $(1-\gamma)L$ elements. Furthermore, we can find these $k$ sets in $O(\log r)$ rounds.
\end{lemma}

\begin{proof}
    Consider the following quantities:
    \begin{align*}
        \phi_1 & = |S_1|\\
        \phi_2 & = |S_1 \cup S_2| - |S_1|\\
        \phi_3 & = |S_1 \cup S_2 \cup S_3| - |S_1 \cup S_2|\\
        \ldots
    \end{align*}
    Clearly, $\sum_{j=1}^r \phi_j = L$.  We say $S_j$ is responsible for element $i$ if $i \in S_j \setminus (\bigcup_{l < j} S_l)$. This establishes a one-to-one correspondence between the sets ${S_1, \dots, S_r}$ and the elements they cover. $S_j$ is responsible for exactly $\phi_j$ elements. Furthermore, if we remove some sets from $\mathcal{S}$, and an element becomes uncovered, the set responsible for that element must have been removed. Thus, if we remove the $\gamma k$ sets corresponding to the $\gamma k$ smallest $\phi_j$, then at most $\gamma L$ elements will not have a responsible set. Thus, the number of elements that become uncovered is at most $\gamma L$. 

    To find these sets, we apply Lemma \ref{lem:prefix-coverage} with $r$ in place of $k$ and $O(\epsilon)$ in place of $\gamma$ to learn about $\phi_1, \phi_2, \ldots, \phi_r$ in $O(\log r)=O(\log k)$ rounds. We then remove the $\gamma k = O(\epsilon)k $ sets corresponding to the $\gamma k$ smallest $\phi_j$ and output the remaining $k$ sets. 
\end{proof}

\paragraph{Putting it all together.} We spend $O(1/\epsilon^3 \cdot \log n \cdot \log m)$ rounds to approximately solve the linear program $\Pi_1$. From there, we can round the solution to find a sub-collection of $k + O(\epsilon)k$ sets that cover at least $(1-1/e-O(\epsilon))\Opt$ elements with high probability in $O(1/\epsilon \cdot \log m)$ rounds. We then apply Lemma \ref{lemma:fraction-of-a-collection} to find $k$ sets among these that cover at least $(1-1/e-O(\epsilon))\Opt$ elements in $O(\log k)$ rounds. The total number of rounds is therefore $O(1/\epsilon^3 \cdot \log n \cdot \log m)$.

\paragraph{Reducing the number of rounds to  $O(1/\epsilon^3 \cdot  \log m  \cdot (\log m+ \log(1/\epsilon) ))$.} The described algorithm runs in $O(1/\epsilon^3 \cdot \log n \cdot \log m)$ rounds. Our main result in Theorem \ref{thm:main} states a stronger bound  $O(1/\epsilon^3 \cdot  \log m \cdot  (\log m+ \log(1/\epsilon) ))$ rounds. We achieve this by adopting the sub-sampling framework of McGregor and Vu \cite{MV19}.

Without loss of generality, we may assume that each element is covered by some set. If not, we can remove all of the elements that are not covered by any set using $O(\log m)$ rounds. Specifically, let $\textbf{v}_j$ be the characteristic vector of $S_j$. We can compute $\textbf{v} = \sum_{i=1}^j \textbf{v}_j$ in $O(\log m)$ rounds using the standard converge-cast binary tree algorithm. We can then remove the elements that are not covered by any set (elements corresponding to 0 entries in $\textbf{v}$).

We now have $m$ sets covering $n$ elements. Since $k = \Omega(m)$, we must have that $\Opt = \Omega(n)$. McGregor and Vu showed that if one samples each element in the universe $[n]$ independently with probability $p = \Theta\left( \frac{\log {m \choose k}} {\epsilon^2 \Opt} \right)$ then with high probability, if we run a $\beta$ approximation algorithm on the subsampled universe, the solution will correpond to a $\beta -\epsilon$ approximation on the original universe. We have just argued that $\Opt = \Omega(n)$ and therefore with high probability, we sample $O \left( \frac{\log  {m \choose k} }{\epsilon^2} \right) = O(1/\epsilon^2 \cdot m)$ elements by appealing to Chernoff bound and the fact that ${m \choose k} \leq 2^m$.

As a result, we may assume that $n = O(1/\epsilon^2 \cdot m)$. This results in an $O(1/\epsilon^2  \cdot \log m \cdot (\log m + \log(1/\epsilon)))$ round algorithm.

\paragraph{Bounded frequency.} Assuming $f = \max_i f_i$ is known, we can lift the assumption that $k = \Omega(m)$ and parameterize our algorithm based on $f$ instead.  McGregor et al. \cite{MTV21} showed that the largest $\roundup{k f/\eta}$ sets contain a solution that covers at least $(1-\eta)\Opt$ elements. We therefore can assume that $m = O(k f/\eta)$ by keeping only the largest $\roundup{k f/\eta}$ sets which can be identified in $O(1)$ rounds. 

 We set  $\epsilon=\eta^2/f$ and proceed to obtain a solution that covers at least $(1-\eta^2/f)(1-\eta)\Opt = (1-O(\eta)) \Opt$ elements using at most $k + O(\epsilon m) = k + O(\eta^2/f \cdot k f/\eta) = k + O(\eta k)$ sets as in the discussion above. Appealing to Lemma \ref{lemma:fraction-of-a-collection}, we can find $k$ sets that cover at least $(1-O(\eta)) \Opt$ elements. The total number of rounds is $O(f^3/\eta^6 \cdot \log \frac{kf}{\eta} \cdot (\log \frac{1}{\eta} + \log  \frac{kf}{\eta})) = O(f^3/\eta^6 \cdot \log
 ^2\frac{kf}{\eta})$.

%%---------------------------------------------------------------------------------------------------------
%%---------------------------------------------------------------------------------------------------------
%%---------------------------------------------------------------------------------------------------------
%%---------------------------------------------------------------------------------------------------------

\subsection{Approximate the LP's solution via multiplicative weights} \label{sec:mwu}

Fix an objective value $L$. Let $P$ be a convex region defined by
    \[
        P = \left\{(\textbf{x},\textbf{z}) \in [0,1]^n \times [0,1]^m: \sum_{i \in [n]} x_i = L \text{ and } \sum_{j \in [m]} z_j = m-k \right\}.    
    \]

Note that if $\left( \textbf{x}_1,\textbf{z}_1), (\textbf{x}_2,\textbf{z}_2 \right), \ldots, \left(\textbf{x}_T,\textbf{z}_T \right) \in P$ then $ \left(\frac{1}{T} \sum_{t=1}^T \textbf{x}_t, \frac{1}{T} \sum_{t=1}^T \textbf{z}_t \right) \in P$. Consider the following problem $\Psi_1$ that asks to  either correctly declare that \[
        \nexists (\textbf{x}, \textbf{z}) \in P: \frac{x_i}{f_i} + \frac{1}{f_i}\cdot\sum_{S_j \ni i} z_j \le 1,  \ \ \ \forall i \in [n]
    \]
or to output a solution $(\textbf{x}, \textbf{z}) \in P$ such that
    \[
        \frac{x_i}{f_i} + \frac{1}{f_i}\cdot\sum_{S_j \ni i} z_j \le 1 + \epsilon \ \ \ \forall i \in [n].
    \]
    
Once we have such an algorithm, we can try different values of $L= \rounddown{(1+\epsilon)^0}, \rounddown{(1+\epsilon)^1}, \rounddown{(1+\epsilon)^2}, \ldots, n$ and return the solution corresponding to the largest $L$ that has a feasible solution. There are $O(1/\epsilon \cdot \log n)$ such guesses. We know that the guess $L$ where $\Opt/(1+\epsilon) \leq L \leq \Opt $ must result in a feasible solution.

To avoid introducing a $\log n$ factor in the number of rounds, we partition these $O(1/\epsilon \cdot \log n)$ guesses into batches of size $O(1/\epsilon)$ where each batch corresponds to $O(\log n)$ guesses. Algorithm copies that correspond to guesses in the same batch will run in parallel. This will only introduce a $\log n$ factor in terms of memory used by each machine. By returning the solution corresponding to the largest feasible guess $L$, one attains Theorem \ref{thm:solving-mkc-lp}.

\paragraph{Oracle implementation.} Given a weight vector $\textbf{w} \in \RR^n$ in which $w_i \ge 0$ for all $i \in [n]$. We first consider an easier feasibility problem $\Psi_2$. It asks to either correctly declares that 
    \begin{align*} 
        \nexists (\textbf{x},\textbf{z}) \in P:\ \  \sum_{i=1}^n w_i  \cdot \left(\frac{x_i}{f_i}+\sum_{S_j \ni i} \frac{z_i}{f_i}\right) \leq \sum_{i=1}^n w_i, \ \ \ \forall i \in [n]
    \end{align*}
or to outputs a solution $(\textbf{x},\textbf{z}) \in P$ such that
    \begin{align} \label{eq:oracle-prob}
        \sum_{i=1}^n w_i  \cdot \left(\frac{x_i}{f_i}+\sum_{S_j \ni i} \frac{z_i}{f_i}\right) \leq \sum_{i=1}^n w_i + \frac{1}{n^{5}}, \ \ \ \forall i \in [n].
    \end{align}

That is, if the input is feasible, then output the corresponding $(\textbf{x},\textbf{z}) \in P$ that approximately satisfy the constraint. Otherwise, correctly conclude that the input is infeasible.  In the multiplicative weights update framework, this is known as the approximate oracle.

Note that if there is a feasible solution to $\Psi_1$, then there is a feasible solution to $\Psi_2$ since
    \begin{align*}
        \frac{x_i}{f_i} + \frac{1}{f_i}\cdot\sum_{S_j \ni i} z_j \le 1 \ \ \ \forall i \in [n] \implies \sum_{i=1}^n w_i  \cdot \left(\frac{x_i}{f_i}+\sum_{S_j \ni i} \frac{z_i}{f_i}\right) \leq \sum_{i=1}^n w_i .
    \end{align*}

We can implement an oracle that solves the above feasibility problem $\Psi_2$ as follows. First, observe that
    \begin{align*}
        \sum_{i=1}^n \frac{w_i}{f_i} \cdot \left(x_i+\sum_{S_j \ni i} z_j\right) \le \sum_{i=1}^n w_i \\
        \iff \sum_{i=1}^n \frac{w_i}{f_i} \cdot x_i + \sum_{i=1}^n \frac{w_i}{f_i} \cdot \sum_{S_j \ni i} z_j \le \sum_{i=1}^n w_i \\ 
        \iff \sum_{i=1}^n \frac{w_i}{f_i} \cdot x_i + \sum_{j=1}^m z_j \cdot \sum_{i \in S_j} \frac{w_i}{f_i} \le \sum_{i=1}^n w_i.
    \end{align*}
To ease the notation, define
    \begin{align*}
        p_i := \frac{w_i}{f_i}, \ \ \ \forall i \in [n] \text{, and } q_j  := \sum_{i \in S_j} \frac{w_i}{f_i} =  \sum_{i \in S_j} p_i, \ \ \ \forall j \in [m].
    \end{align*}
We therefore want to check if there exists $(\textbf{x},\textbf{z}) \in P$ such that
    \begin{align*}
        LHS(\textbf{x}, \textbf{z}) := \sum_{i=1}^n x_i  p_i + \sum_{j=1}^m z_j  q_j \le \sum_{i=1}^n w_i.
    \end{align*}
We will minimize the left hand side by minimizing each sum separately.  We can indeed do this exactly. However, there is a subtle issue where we need to bound the number of bits required to represent  $p_i = \frac{w_i}{f_i}$ and  $q_j = \sum_{i \in S_j} \frac{w_i}{f_i} = \sum_{i \in S_j} p_i$ given the memory constraint. To do this, we truncate the value of each $p_i$ after the $(10 \log_2 n)$-th bit following the decimal point. Note that this will result in an underestimate of $p_i$ by at most $1/n^{10}$. In particular, let $\hat{p}_i$ be $p_i$ after truncating the value of $p_i$ at the $(10 \log_2 n)$-th bit after the decimal point and $\hat{q}_j = \sum_{i \in S_j} \hat{p}_i$. For any $(\textbf{x}, \textbf{z}) \in [0,1]^n \times [0,1]^m$, we have
\begin{align*}
    \widehat{LHS}(\textbf{x}, \textbf{z}) & := \sum_{i=1}^n \hat{p}_i  x_i + \sum_{j=1}^m z_j  \sum_{i \in S_j} \hat{p}_i \geq \left( \sum_{i=1}^n p_i  x_i \right) - \frac{n}{n^{10}} + \left(\sum_{j=1}^m z_j  \sum_{i \in S_j} p_i \right) - \frac{m f_i}{n^{10}} \\
    & >  LHS(\textbf{x}, \textbf{z}) - \frac{1}{n^5}.
\end{align*}

The last inequality is based on the assumption that $m, f_i \le n$. Therefore, $LHS(\textbf{x}, \textbf{z}) - 1/n^5 \le \widehat{LHS}(\textbf{x}, \textbf{z}) \le LHS(\textbf{x}, \textbf{z})$. Note that since $\sum_{i=1}^n x_i = L$ and $\sum_{j=1}^m z_j = m-k$ , to minimize $\widehat{LHS}(\textbf{x}, \textbf{z})$ over $(\textbf{x}, \textbf{z}) \in P$, we simply set
    \begin{align*}
        x_i & = [\text{$\hat{p}_i$ is among the $L$ smallest values of $\{ \hat{p}_t \}_{t=1}^n$} \}], \\
        z_j & = [\text{$\hat{q}_j$ is among the $m-k$ smallest values of $\{ \hat{q}_t \}_{t=1}^m$} \}].
    \end{align*}

After setting $\textbf{x}, \textbf{z}$ as above, if $\widehat{LHS}(\textbf{x}, \textbf{z}) > \sum_{i = 1}^n w_i \implies LHS(\textbf{x}, \textbf{z}) > \sum_{i = 1}^n w_i$, then it is safe to declare that the system is infeasible. Otherwise, we have found $(\textbf{x}, \textbf{z}) \in P$ such that $\widehat{LHS}(\textbf{x}, \textbf{z}) \le \sum_{i = 1}^n w_i \implies LHS(\textbf{x}, \textbf{z}) \le \sum_{i = 1}^n w_i + 1/n^5$ as required by Equation (\ref{eq:oracle-prob}).

\begin{lemma} \label{lem:oracle}
    Assume that all machines have the vector $\textbf{w}$. We can solve the feasibility problem $\Psi_2$ in $O(1)$ rounds, where all machines either learn that the system is infeasible or obtain an approximate solution $(\textbf{x}, \textbf{z}) \in P$ that satisfies Equation (\ref{eq:oracle-prob}).
\end{lemma}

\begin{proof}
	We have argued above that the oracle we design either correctly declares that the system is infeasible or outputs $(\textbf{x}, \textbf{z}) \in P$ such that $LHS(\textbf{x}, \textbf{z}) \le \sum_{i = 1}^n w_i + 1/n^5$.
	Recall that each machine has the frequency vector $\textbf{f}$ and suppose for the time being that each machine also has the vector $\textbf{w}$.  We can implement this oracle in $O(1)$ rounds as follows. Each machine $j$ computes $\hat{q}_j = \sum_{i \in S_j} \hat{p}_i$ since it has $w_i$ and $f_i$ to compute $p_i$ for all $i \in [n]$. Then it sends $\hat{q}_j$ to the central machine. Note that each $\hat{p}_i$ can be represented using $O(\log n)$ bits. Observe that $\hat{q}_j$ is the sum of at most $n$ different $\hat{p}_i$. The number of bits in the fractional part does not increase while the number of bits in the integer part increases by at most $O(1)$, as the sum is upper bounded by a crude bound $2n \times 2^{O(\log n)} = 2^{O(\log n)}$. Thus, the central machine receives at most $\tO{m} = \tO{n}$ bits from all other machines.

	The central machine will then set $\textbf{x}$ and $\textbf{z}$ as described above. This allows the central machine to correctly determine whether the system is infeasible or to find $(\textbf{x},\textbf{z}) \in P$ that satisfies Equation (\ref{eq:oracle-prob}).   Since the entries of $\textbf{x}$ and $\textbf{z}$ are binary, they can be sent to non-central machines, with each machine receiving a message of $n + m = O(n)$ bits. We summarize the above as the following lemma.
\end{proof}

\paragraph{Solving the LP via multiplicative weights.}  Once the existence of such an oracle is guaranteed, we can follow the multiplicative weights framework to approximately solve the LP.  We will first explain how to implement the MWU algorithm in the MPC model. See Algorithm \ref{alg:mwu}.

\begin{algorithm}
    \caption{Multiplicative weights for solving the LP} \label{alg:mwu}
    \KwIn{Objective value $L$, $\epsilon \leq 1/4$}
    Initialize $w_i^{(0)} = 1$ for all $i \in [n]$. \\
    \For{iteration $t=1,2,\ldots,T = O(1/\epsilon^2 \cdot \log n)$} {
        Run the oracle in Lemma \ref{lem:oracle}  with $\textbf{w}^{(t-1)}$ to check if there exists a feasible solution. If the answer is \textsc{Infeasible}, stop the algorithm. If the answer is \textsc{Feasible}, let $\textbf{x}^{(t)}$ and $\textbf{z}^{(t)}$ be the output of the oracle that are now stored in all machines \label{mwu:oracle-step}.

        Each machine $j$ constructs $Y_j = \{Y_{j1},Y_{j2},\ldots,Y_{jn}\}$ where 
            $Y_{ji} = {z_j^{(t)}} \cdot [\{i \in S_j\}]$. \\
        
        Compute $W = \sum_{j \in [m]} Y_j$ in $O(\log m)$ rounds using the a converge-cast binary tree and send $W$ to the central machine. Note that $W_i = \sum_{S_j \ni i} {z_j^{(t)}}$.\\
        For each $i \in [n]$, the central machine computes $E_i^{(t)} = f_i \cdot \textup{error}_i^{(t)}$ where
        \[
            \textup{error}_i^{(t)} := 1 - \frac{x_i^{(t)}}{f_i} - \frac{W_i}{f_i}  = 1 - \frac{x_i^{(t)}}{f_i} - \sum_{S_j \ni i} \frac{z_j^{(t)}}{f_i}  \ \ \ \forall i \in [n].
        \]
        and sends $\sum_{d=1}^t E_i^{(d)}$ to all other machines. \label{MWU:step_w}\\
        For each $i \in [n]$, each machine computes  
	\[
		w_{i}^{(t)}  = 2^{-\epsilon \cdot \sum_{d=1}^t E_i^{(d)} /f_i} =  2^{-\epsilon \cdot \sum_{d=1}^t \textup{error}_i^{(d)}} =  2^{-\epsilon \cdot \textup{error}_i^{(t)}} w_i^{(t-1)}.
	\]

    }

    After $T$ iterations, output 
	\[
		\textbf{x} = \frac{1}{T} \sum_{t=1}^T \textbf{x}^{(t)} \text{, and }  \textbf{z} = \frac{1}{T} \sum_{t=1}^T \textbf{z}^{(t)}.
	\]

\end{algorithm}

\begin{lemma}
    Algorithm \ref{alg:mwu} can be implemented in $O(1/\epsilon^2 \cdot \log n \cdot \log m)$ rounds.
\end{lemma}
\begin{proof}
    Without loss of generality, we can round $\epsilon$ down to the nearest power of $1/2$. Then, $-\epsilon$ can be represented using $O(\log (1/\epsilon)) = O(\log n)$ bits, since we assume $1/\epsilon < n$.

    Each machine maintains the current vectors $\textbf{x}^{(t)}$, $\textbf{z}^{(t)}$, and $\textbf{w}^{(t)}$, each of which has at most $O(n)$ entries. Since the entries in $\textbf{x}^{(t)}$ and $\textbf{z}^{(t)}$ are binary, we can represent them using $O(n)$ bits. Recall that, in the oracle described in Lemma \ref{lem:oracle}, the central machine can broadcast $\textbf{x}^{(t)}$ and $\textbf{z}^{(t)}$ to all other machines in one round without violating the memory constraint.
    
    It remains to show that the number of bits required to represent $w_i^{(t)}$ is $O(\log n)$ in all iterations.  The central machine then implicitly broadcasts $w_i^{(t)}$ to all other machines in iteration $t$ by sending $\sum_{d=1}^t E_i^{(d)}$ 
    to all other machines for each $i \in [n]$. 

Note that each $w_i^{(t)}$ has the form 
    \[ 
        w_i^{(t)} = 2^{-\epsilon \cdot \sum_{d=1}^t E_i^{(d)}/f_i} = 2^{-\epsilon \cdot \sum_{d=1}^t \textup{error}_i^{(d)}}.
    \] 

   For any $d$, since $m,f_i \le n$, we have
    \[
        f_i-m-1 \le f_i \cdot \textup{error}_i^{(d)}  \le f_i \implies f_i \cdot \textup{error}_i^{(d)}  \in [-2n, 2n].
    \]
    Since $T = \Theta(1/\epsilon^2 \cdot \log n) = o(n^3)$, we have
    \[
        \sum_{d=1}^t E_i^{(d)} = \sum_{d=1}^t f_i \cdot \textup{error}_i^{(d)} \in [-2 n t, 2 n t] \subset [-2n^4, 2n^4].
    \]
    Thus, we can represent $\sum_{d=1}^t E_i^{(d)}$ using $O(\log n)$ bits. This implies that the central machine can implicitly broadcast $w_i^{(t)}$ to all other machines in $O(1)$ rounds. Putting it all together, each of the $T = \Theta(1/\epsilon^2 \cdot \log n)$ iterations consists of: 1) a call to the oracle, which takes $O(1)$ rounds, 2) computing $\textup{error}_i^{(t)}$ for each $i \in [n]$ in step \ref{MWU:step_w}, which takes $O(\log m)$ rounds, and 3) the central machine broadcasting $\textbf{w}^{(t)}$ in one round to all other machines before proceeding to the next iteration.
\end{proof}

The next lemma is an adaptation of the standard multiplicative weights algorithm.

\begin{lemma}
    The output of Algorithm \ref{alg:mwu} satisfies the following property. If there exists a feasible solution, then the output satisfies:
        \begin{align*}
            \frac{x_i}{f_i} + \sum_{S_j \ni i} \frac{z_j}{f_i} & \le 1 + \epsilon \ \ \ \forall i \in [n] \text{, and } \sum_{i=1}^n x_i = L.
        \end{align*}
        Otherwise, the algorithm correctly concludes that the system is infeasible.
\end{lemma}

\begin{proof}
    If the algorithm does not output \textsc{Infeasible}, this implies that in each iteration $t$, $\sum_{i=1}^n x_i^{(t)} = L$, and therefore $\sum_{i=1}^n x_i = L$ since $\textbf{x} = \frac{1}{T} \sum_{t=1}^T \textbf{x}^{(t)}$. Similarly, in each iteration $t$, $\sum_{j=1}^m z_j^{(t)} = m-k$, and thus $\sum_{j=1}^m z_j = m-k$ since $\textbf{z} = \frac{1}{T} \sum_{t=1}^T \textbf{z}^{(t)}$. Hence, the output $(\textbf{x}, \textbf{z}) \in P$. Define the potential function
    \begin{align*}
        \Phi^{(t)} := \sum_{i=1}^n w_i^{(t)}.
    \end{align*}

    We will make use of the fact $\exp{-\eta x} \leq 1 - \eta x + \eta^2 x^2$ for $|\eta x| \leq 1$. Note that for all $i$ and $t$, we always have
    \[
        \textup{error}_i^{(t)} =  1 - \frac{x_i^{(t)}}{f_i} - \sum_{S_j \ni i} \frac{z_j^{(t)}}{f_i}   \in [-1, 1].
    \]
    Note that $\textup{error}_i^{(t)} \in [-1, 1]$ because $0 \leq x_i^{(t)} \leq 1 \leq f_i$ and $0 \leq \sum_{S_j \ni i} {z_j^{(t)}} \leq \sum_{S_j \ni i} 1 = f_i $. Set $\alpha = \epsilon\cdot \ln(2)$, as long as $\epsilon \le 1/4$, we have $|\alpha \cdot \textup{error}_i^{(t)} | < 1$ and therefore
    \begin{align*}
        w_i^{(t)} & = \exp{-\alpha\cdot \textup{error}_i^{(t)}} \cdot w_i^{(t-1)} \leq \left( 1 - \alpha \cdot \textup{error}_i^{(t)} + \alpha^2 \cdot \left(\textup{error}_i^{(t)}\right)^2 \right) \cdot w_i^{(t-1)}.
    \end{align*}
    Summing over $i$ gives:
    \begin{align*}
        \Phi^{(t)} & \leq \sum_{i=1}^n \left( 1 - \alpha \cdot \textup{error}_i^{(t)} +  \alpha^2 \right) \cdot w_i^{(t-1)} \\
            & = (1+\alpha^2) \sum_{i=1}^n w_i^{(t-1)} - \alpha \sum_{i=1}^n \textup{error}_i^{(t)} \cdot w_i^{(t-1)}.
    \end{align*}
    The first inequality follows from the fact that $\left(\textup{error}_i^{(t)}\right)^2 \in [0, 1]$.
    Note that
    \begin{align*}
        \sum_{i=1}^n \textup{error}_i^{(t)} w_i^{(t-1)} &= \sum_{i=1}^n w_i^{(t-1)} \left( 1 - \frac{x_i^{(t-1)}}{f_i} - \sum_{S_j \ni i} \frac{z_j^{(t-1)}}{f_i} \right) \\
        & = \sum_{i=1}^n w_i^{(t-1)} - \sum_{i=1}^n w_{i}^{(t-1)} \left( \frac{x_i^{(t-1)}}{f_i}  + \sum_{S_j \ni i} \frac{z_j^{(t-1)}}{f_i} \right) \geq -\frac{1}{n^5}.
    \end{align*}
    The last inequality follows from the fact that the oracle is guaranteed to find $\textbf{x}^{(t-1)}$ and $\textbf{z}^{(t-1)}$ such that
    \begin{align*}
        \sum_{i=1}^n \left( \frac{x_i^{(t-1)}}{f_i} +  \sum_{S_j \ni i}\frac{z_j^{(t-1)}}{f_i}  \right) w_i^{(t-1)} \leq \sum_{i=1}^n w_i^{(t-1)} + \frac{1}{n^5}.
    \end{align*}
    
    Thus, $\Phi^{(t)} \leq (1+  \alpha^2) \sum_{i=1}^n w_i^{(t-1)} + \frac{\alpha}{n^5} = (1+\alpha^2) \Phi^{t-1} + \frac{\alpha}{n^5}.$ By a simple induction,
    \begin{align*}
        \Phi^{(T)} & \le (1 + \alpha^2) \Phi^{(T-1)} + \frac{\alpha}{n^5} \\
        & \le (1 + \alpha^2)^2 \Phi^{(T-2)} + \frac{\alpha}{n^5} \left(1 + \alpha^2\right) + \frac{\alpha}{n^5} \\     
        & \le (1 + \alpha^2)^3 \Phi^{(T-3)} + (1 + \alpha^2)^2 \frac{\alpha}{n^5} + \frac{\alpha}{n^5} \left(1 + \alpha^2\right) + \frac{\alpha}{n^5} \\
        & \ldots \\ 
        & \le (1 + \alpha^2)^T \Phi^{(0)} + \frac{\alpha}{n^5} \sum_{t=0}^{T-1} (1 + \alpha^2)^t  \\
        & =  (1 + \alpha^2)^T \Phi^{(0)} + \frac{\alpha}{n^5}  \cdot \frac{(1 + \alpha^2)^T - 1}{\alpha^2} \\
        & \leq  (1 + \alpha^2)^T \Phi^{(0)} + \frac{1}{\alpha n^5}  {(1 + \alpha^2)^T } .
    \end{align*}
    Recall that $\alpha = \epsilon \ln2$ and we assume $1/\epsilon < n/10$. Thus, $1/(\alpha n^5) < \epsilon^4/ \ln 2 <1$. Furthermore, recall that $\Phi^{(0)} = n$. We have, 
    \begin{align*}
        w_i^{(T)} &  \leq (1 + \alpha^2)^T (n+1) < (1 + \alpha^2)^T 2n \\
        \exp{-\alpha \sum_{t=1}^T \textup{error}_i^{(t)}} & \leq (1  + \alpha^2)^T (2n)  \\
        -\alpha \sum_{t=1}^T \textup{error}_i^{(t)} & \leq \ln (2n) + T \ln (1 + \alpha^2).
    \end{align*}
    We use the fact that $\ln(1+x) \le x$ for $x \in \RR$ to get 
    \begin{align*}
        \sum_{t=1}^T \textup{error}_i^{(t)} & \geq - \frac{\ln (2 n)}{\alpha} - T \frac{\ln (1  +  \alpha^2)}{\alpha} \\
        \sum_{t=1}^T \left( 1 - \frac{x_i^{(t)}}{f_i} - \sum_{S_j \ni i} \frac{z_j^{(t)}}{f_i} \right) 
        & \geq - \frac{\ln (2 n)}{\alpha} - T \frac{ \alpha^2 }{\alpha}  \\
        \frac{1}{T}\sum_{t=1}^T \left( 1 - \frac{x_i^{(t)}}{f_i} - \sum_{S_j \ni i} \frac{z_j^{(t)}}{f_i} \right) 
        & \geq - \frac{\ln (2 n)}{T \alpha} - \alpha \\
        1 - \frac{x_i}{f_i} - \sum_{S_j \ni i} \frac{z_j}{f_i} & \geq - \frac{\ln (2n)}{T \alpha} -  \alpha \\
        \frac{x_i}{f_i} + \sum_{S_j \ni i} \frac{z_j}{f_i} & \leq 1 + \frac{\ln (2n)}{T \alpha} +  \alpha \\
        \frac{x_i}{f_i} + \sum_{S_j \ni i} \frac{z_j}{f_i} & \leq 1 + O(\epsilon).
    \end{align*}

    The last inequality follows from choosing $T = \Theta(1/\epsilon^2 \cdot \log n)$ and the fact that $\alpha = \epsilon \ln(2)$; furthermore, recall that the final solution $x_i = \frac{1}{T} \sum_{t=1}^T x_i^{(t)}$ and $z_j = \frac{1}{T} \sum_{t=1}^T z_j^{(t)}$. Thus, the output of the algorithm satisfies the desired properties.

\end{proof}

\bibliographystyle{plain}
\bibliography{ref}

\appendix

\section{Omitted Proofs} \label{sec:omitted-proofs}

\begin{proof}[Proof of Lemma \ref{lem:rounding}]
    Interpret $\{y_1/k, y_2/k, \ldots, y_m/k\}$ as the probabilities for the sets $S_1, S_2, \ldots, S_m$. The central machine samples $k$ sets independently according to this distribution. In expectation, we cover at least $(1 - 1/e)L$ elements (see \cite{Steurer}). Hence, in expectation, the number of uncovered elements is at most $L/e$. By Markov's inequality, the probability that the number of uncovered elements is more than $(1 + \epsilon)L/e$ is at most $1/(1 + \epsilon) \leq e^{-\epsilon/2}$. So, if we do this $\Theta(1/\epsilon \cdot \log m)$ times, the probability that the best solution covers fewer than $(1 - 1/e - 1/(1 + \epsilon))L$ elements is at most 
    \[
        \exp{-\epsilon/2 \cdot \Theta(1/\epsilon \cdot \log m)} \leq 1/\poly(m).
    \]

    After the central machine forms $\Theta(1/\epsilon \cdot \log m)$ solutions (each solution consisting of at most $k$ sets) as described above, it broadcasts these solutions in batches of size $1/\epsilon$ to all other machines. Note that each batch contains $O(\log m)$ solutions.

    For each batch, each machine receives a message of size $\tO{k}$. For each of these $O(\log m)$ solutions in the batch, the machines can compute its coverage as follows. Let $v_j$ be the characteristic vector of the set $S_j$ that was chosen. We can aggregate the vectors $v_j$ that are in the solution in $O(\log m)$ rounds and count the number of non-zero entries using a binary broadcast tree. We do this in parallel for all solutions in the batch. Finally, we repeat this process for all $1/\epsilon$ batches.
\end{proof}

% \begin{algorithm}[H] \label{alg:oracle}
%     \caption{Oracle implementation} \label{alg:oracle}
%     \KwIn{A weight vector $\textbf{w} \in \RR^n$ that is shared among all machines}

%     In parallel, for all $j \in [m]$, machine $j$ computes $q_j = \sum_{i \in S_j} {w_i}/f_i$  and send this to the central machine. % This allows the central machines to compute $q_j = \sum_{i \in S_j} w_i / f_j$.

%     The central machine sets $z_j = 1$ if $q_j$ is among the $m-k$ smallest values of $\{ q_l \}_{l=1}^m$ and sets $z_j = 0$ otherwise. \\
    
%     Similarly, it sets $x_i = 1$ for if $p_i$ is among the $L$ smallest values of $\{ p_l \}_{l=1}^n$ and sets $x_i = 0$ otherwise. \\

%     If $\sum_{i=1}^n x_i \cdot p_i + \sum_{j=1}^m z_j \cdot q_j \le \sum_{i=1}^n w_i$, output $\textbf{x}$ and $\textbf{z}$. Otherwise, output \textsc{Infeasible}.
% \end{algorithm}

% \begin{proof}[Proof of Lemma \ref{lem:oracle}]
%     Consider Algorithm \ref{alg:oracle}. Each machine $j$ sends  $\sum_{i \in S_j} w_i$ to the central machine. Because we assume that representing each $w_i$ requires $C \log n$ bits for some constant $C$, we know that $0 \leq w_i < n^{C}$. Therefore, $0 \leq \sum_{i \in S_j} w_i < n^{C + 1}$ and the total message sent by non-central machines to the central machine has size at most $O(\log n)$ bits. Hence, the central machine does not receive more than $O(n \log n)$ bits from all other machines.
% \end{proof}

\end{document}